\documentclass[11pt]{amsart}
\usepackage{graphicx}
\usepackage{amsmath}
\usepackage{amssymb}
\usepackage{amsthm}
\usepackage{eufrak}

\usepackage{fullpage}

\newtheorem{theorem}{Theorem}[section]

\newtheorem{lemma}[theorem]{Lemma}
\newtheorem{proposition}[theorem]{Proposition}
\theoremstyle{definition}

\theoremstyle{remark}

\theoremstyle{definition}

\numberwithin{equation}{section}

\newcommand{\set}[1]{\left\{#1\right\}}

\newcommand{\su}[1]{\mathfrak{su}(#1)}
\newcommand{\ualg}[1]{\mathfrak{u}(#1)}
\newcommand{\gau}{\mathfrak{t}}
\newcommand{\matrdpd}[4]{\left( \begin{array}{cc} #1 & #2 \\[0.2cm]  #3 & #4 \end{array} \right)}

\newcommand{\R}{\mathbb R}
\newcommand{\Z}{\mathbb Z}

\newcommand{\C}{\mathbb C}
\newcommand{\HH}{\mathbb H}
\newcommand{\A}{\mathbf A}
\newcommand{\SU}{\mathbf S}
\newcommand{\eps}{\varepsilon}
\newcommand{\lie}{\mathfrak{g}}
\newcommand{\LL}{{\mathcal L}}

\begin{document}

\title[Solitary waves and vortices in non-abelian gauge theories with matter]{Solitary waves and vortices in non-Abelian gauge theories with matter}%
\author{Vieri Benci}%
\address{ Dipartimento di Matematica Applicata, Universit\`a di Pisa, via F. Buonarroti n.1/c, 56127 Pisa, Italy\newline \indent Department of Mathematics, College of Science, King Saud University, Riyadh, 11451, Saudi Arabia}
\email{benci@dma.unipi.it}%
\author{Claudio Bonanno}%
\address{Dipartimento di Matematica Applicata, Universit\`a di Pisa, via F. Buonarroti n.1/c, 56127 Pisa, Italy}%
\email{bonanno@mail.dm.unipi.it}%
\thanks{This work was partially supported by King Saud University, Riyadh. }
\begin{abstract}
We consider a non-Abelian gauge theory in $\R^{4}$ equipped with the Minkowski metric, which provides a model for the interaction between a bosonic matter field and a gauge field with gauge group $SU(2)$. We prove the existence of solitary waves which are related to those found for the Klein-Gordon-Maxwell equations.
\end{abstract}
\maketitle
\section{Introduction} \label{intro}

We study the existence of \textit{solitary waves},  that is solutions with localised energy and charge, for equations modelling interactions between matter and a gauge field in $\R^{4}$ equipped with the Minkowski metric. When matter is given by the fundamental representation of a field with only one internal symmetry, that is $\psi: \R^{4}\to \C$ is the fundamental representation of $U(1)$, then one speaks of \textit{Abelian gauge theories}. The existence of \textit{solitary waves} in Abelian gauge theories, and in particular for the nonlinear Klein-Gordon-Maxwell equations, has been studied in \cite{coleman} under the name of \textit{charged} or \textit{gauged Q-balls}, and more recently in \cite{bf}, \cite{bf-vort} and \cite{long}. In particular in \cite{bf-vort} it has been proved the existence of \textit{spinning Q-balls}, that is solitary waves with non-vanishing angular momentum, also called \textit{vortices}. In all the cited papers the solitary waves are found by using a particular ansatz for the solution, namely 
$$\psi(t,x)= u(x) e^{\imath S(t,x)}\in \C.$$
In this paper we prove existence of solitary waves and vortices for non-Abelian gauge theories, in particular in the case of a matter field with three internal symmetries, namely with $\psi: \R^{4}\to \C^{2}$ being the fundamental representation of the gauge group $SU(2)$. The part of the Lagrangian density relative to $\psi$ is the Lagrangian of the nonlinear wave equation with covariant derivatives and it models a bosonic field. 

In our model the nonlinear term $W(s)$, which models the self-interaction of the matter field, is not of Higgs-type, but it is non-negative and vanishes only at the origin. This choice is justified by the fact that in the Abelian case the following holds (\cite{bf-arx} and Theorem \ref{no-higgs} below, and \cite{bf-vort}): 
\begin{itemize}
\item if $W(0)>0$ there are not finite energy solutions;
\item even in the case in which $W(s)$ is convex there is a symmetry breaking in the sense that there are solutions, not invariant under the action of the full gauge group, which minimize the energy constrained by some other conserved quantity.
\end{itemize}  

In Section \ref{gt}, we recall the basic notions of gauge theories in the Minkowski space. Then, in Section \ref{sec:nonab-case}, we develop the theory in the $SU(2)$ case, writing the matter field in polar form using the exponential map and we give our existence result.

\section{Gauge theories in $\R^{4}$} \label{gt}

Let $G$ be a subgroup of $U(N)$, the unitary group in $\C^{N}$, and
denote by $\Lambda^{k}(\mathbb{R}^{4},\lie)$ the set of $k$-forms
defined in $\mathbb{R}^{4}$ with values in the Lie algebra $\mathfrak{g}$ of
the group $G$. A 1-form 
$$
\Gamma =\sum_{j=0}^{3}\Gamma _{j}dx^{j}\in \Lambda ^{1}(\mathbb{R}^{4}, \mathfrak{g})
$$
is called \textit{connection form}. The \textit{matter field} $\psi$ is a smooth $\C^{N}$-valued function
$$\psi : \R^{4} \to \C^{N}$$
and by using a connection $\Gamma$ we consider the operator
$$
d_{\Gamma }=\Lambda ^{0}(\mathbb{R}^{4},\C^{N})\rightarrow \Lambda
^{1}(\mathbb{R}^{4},\C^{N})
$$
called \textit{covariant differential} and defined by 
$$
d_{\Gamma }=d+q\, \Gamma =\sum_{j=0}^{3}\left( \partial _{j}+q \Gamma _{j}\right)
dx^{j}
$$
where $q$ is a real positive parameter. The operators 
$$
D_{j}=\partial _{j}+ q\Gamma _{j}:\mathcal{C}^{1}\left( \mathbb{R}^{4},\mathbb{
C}^{N}\right) \rightarrow \mathcal{C}^{0}\left( \mathbb{R}^{4},\mathbb{C}^{N}\right) ,\qquad j=0,\dots,3
$$
are called \textit{covariant derivatives}. Hence
\begin{equation} \label{diff-for-matter}
d_{\Gamma }\psi =\sum_{j=0}^{3}\, D_{j}\psi\, dx^{j} = \sum_{j=0}^{3}\, \left( \partial_{j} \psi + q \Gamma_{j} \psi \right)\, dx^{j}
\end{equation}
where by $\Gamma_{j} \psi$ we denote the action of an element of $\lie$ on $\psi$ by the usual matrix representation. The parameter $q$ plays the role of the coupling constant between the matter field and the connection $\Gamma$.

The covariant differential can be extended to $k$-forms $\Lambda^{k}(\mathbb{R}^{4},\C^{N})$ by letting
$$
d_{\Gamma}(\psi \otimes \omega) = d_{\Gamma}\psi \otimes \omega + \psi\otimes d_{\Gamma}\omega\ \in \Lambda^{k+1}(\mathbb{R}^{4},\C^{N}) \qquad \forall\, \psi \in \mathcal{C}^{1}\left( \mathbb{R}^{4},\mathbb{C}^{N}\right),\ \omega \in \Lambda^{k-1}(\mathbb{R}^{4},\C^{N})
$$
but in general the operator
$$
d_{\Gamma }\circ d_{\Gamma }: \Lambda ^{0}(\mathbb{R}^{4},\C^{N}) \rightarrow \Lambda ^{2}(\mathbb{R}^{4},\C^{N})
$$
does not vanish, in fact 
$$
\left( d_{\Gamma }\circ d_{\Gamma }\right) \psi =\frac q2\, \sum_{j,k=0}^{3}\left( \partial _{k}\Gamma _{j}-\partial _{j}\Gamma _{k}+ 
q \left[ \Gamma _{k},\Gamma _{j}\right] \right) \psi \ dx^{k}\wedge dx^{j}
$$
and the 2-form 
\begin{equation} \label{curvature}
F_{\Gamma }=\frac q2\, \sum_{j,k=0}^{3} F_{kj} dx^{k}\wedge dx^{j}\ \in \Lambda^{2}(\mathbb{R}^{4},\lie), \qquad F_{kj}:= \partial _{k}\Gamma _{j}-\partial_{j}\Gamma _{k}+q \left[ \Gamma _{k},\Gamma _{j}\right] \in \lie
\end{equation}
is called \textit{curvature}. We now extend the covariant differential $d_{\Gamma}$ to $k$-forms with values in $\lie$. For two 1-forms in $\Lambda ^{1}(\mathbb{R}^{4},\lie)$
\begin{equation} \label{esempi}
A = \sum_{j=0}^{3}\,A_{j}\,dx^{j} \qquad B =\sum_{j=0}^{3}\,B_{j}\,dx^{j}
\end{equation}
one can define
\begin{equation*}
A\wedge B:=\sum_{i,j=0}^{3}\,A_{i}B_{j}\,dx^{i}\wedge dx^{j}\ \in \Lambda ^{2}(
\mathbb{R}^{4},\mathfrak{g})
\end{equation*}
and
\begin{equation}
[A,B]:=\sum_{i,j=0}^{3}\,[A_{i},B_{j}]\,dx^{i}\wedge dx^{j} \ \in \Lambda ^{2}(
\mathbb{R}^{4},\mathfrak{g})
\end{equation}
Notice that $[A,B] = A\wedge B + B\wedge A$, hence in particular $[A,A] = 2A\wedge A$. By using these operators, the covariant differential can be extended to an operator acting on $0$-forms with values in $\lie$
$$
d_{\Gamma }=\Lambda ^{0}(\mathbb{R}^{4},\mathfrak{g})\rightarrow \Lambda
^{1}(\mathbb{R}^{4},\mathfrak{g})
$$
by the same definition
$$
d_{\Gamma} \alpha =\sum_{j=0}^{3}\left( \partial _{j} \alpha +q \, [\Gamma_{j},\alpha] \right) dx^{j} \qquad \forall\, \alpha \in {\mathcal C}^{1}(\R^{4}, \lie)
$$
letting
\begin{equation} \label{der-cov-gauge}
D_{j} : \mathcal{C}^{1}\left( \mathbb{R}^{4},\lie \right) \rightarrow \mathcal{C}^{0}\left( \mathbb{R}^{4},\lie \right), \qquad D_{j}\alpha = \partial_{j}\alpha +q \, [\Gamma_{j},\alpha], \quad j=0,\dots,3
\end{equation}
Then $d_{\Gamma}$ can be extended to $\Lambda^{k}(\mathbb{R}^{4},\mathfrak{g})$ as above. For example if $A \in \Lambda ^{1}(\mathbb{R}^{4},\lie)$ then
\begin{equation*}
d_{\Gamma }A=dA+q [\Gamma ,A]=\sum_{i,j=0}^{3}\,\left( \partial
_{i}A_{j}+q[\Gamma _{i},A_{j}]\right) \,dx^{i}\wedge dx^{j}
\end{equation*}
Using this notation we have
$$F_{\Gamma} = d\Gamma + \Gamma \wedge \Gamma$$

Now we need to define a scalar product on $\Lambda^{k}(\mathbb{R}^{4}, \mathfrak{g})$. First we recall the definition of the Hilbert product on $\mathfrak{g}$ given by
\begin{equation} \label{hilb}
\left\langle U,V\right\rangle =Tr(U^{\ast }V), \qquad \| U \|^{2} := \left\langle U,U\right\rangle 
\end{equation}
where $U^{\ast }$ denote the adjoint of $U$. Next, we equip $\mathbb{R}^{4}$ with the Minkowski quadratic form given by 
\begin{equation*}
< v,v>_{M}=-\left\vert v_{0}\right\vert
^{2}+\sum_{j=1}^{3}\left\vert v_{j}\right\vert ^{2}.
\end{equation*}
where $v=(v_{0},v_{1},v_{2},v_{3})$ is a 4-vector. The Minkowski quadratic form can be extended to the space of the differential forms $\alpha \in \Lambda ^{k}(\mathbb{R}^{4}),$ and it will be denoted by 
\begin{equation} \label{mink}
< \alpha ,\alpha >_{M}.  
\end{equation}
For future reference we recall the relations 
\begin{equation}
<dx^{i},dx^{j}>_{M}=0\quad i\not=j\qquad <dx^{0},dx^{0}>_{M}=-1\qquad
<dx^{i},dx^{i}>_{M}=1\quad i=1,2,3  \label{come-1-forme}
\end{equation}
\begin{equation}
<dx^{0}\wedge dx^{i},dx^{0}\wedge dx^{i}>_{M}=-1\quad i=1,2,3\qquad
<dx^{i}\wedge dx^{j},dx^{i}\wedge dx^{j}>_{M}=1\quad i,j=1,2,3
\label{come-2-forme}
\end{equation}
and the products vanish in the other cases. 

Now, since $\Lambda ^{k}(\mathbb{R}^{4},\C^{N})=\C^{N} \otimes \Lambda
^{k}(\mathbb{R}^{4})$ and $\Lambda ^{k}(\mathbb{R}^{4},\mathfrak{g})=\mathfrak{g}\otimes \Lambda^{k}(\mathbb{R}^{4})$, let us define the Minkowsky product on $\Lambda ^{k}(\mathbb{R}^{4},\C^{N})$ as
\begin{equation} \label{prod-forme-c}
<v \otimes \omega ,w \otimes \nu >_{M}:=\, (v,w)_{\C^{N}} <\omega ,\nu >_{M}
\end{equation}
and on $\Lambda ^{k}(\mathbb{R}^{4},\mathfrak{g})$ as
\begin{equation} \label{prod-forme-e}
<U\otimes \omega ,V\otimes \nu >_{M}:=\left\langle U,V\right\rangle <\omega ,\nu >_{M}
\end{equation}
In particular, given two 1-forms as in (\ref{esempi}), we have that
\begin{equation}
<A,B>_{M}\ =\sum_{i,j=0}^{3}\,tr(A_{j}^{\ast }\,B_{j})\,<dx^{i},dx^{j}>_{M}
\label{prod-forme-g}
\end{equation}
and using (\ref{hilb}) and (\ref{come-1-forme})
\begin{equation} \label{auto-prod-1f}
<A,A>_{M}=-\left\Vert A_{0}\right\Vert ^{2}+\sum_{j=1}^{3}\left\Vert
A_{j}\right\Vert ^{2}
\end{equation}

We are now ready to define the Lagrangian density of a gauge theory depending on the matter field $\psi$ and on the connection $\Gamma$. Recall that the parameter $q$ introduced in the definition of the covariant derivatives play the role of the coupling constant between the two fields. Set 
\begin{equation} \label{elle0}
\mathcal{L}_{0}(\psi,\Gamma) :=- \frac{1}{2}< d_{\Gamma }\psi
,d_{\Gamma }\psi >_{M} = \frac{1}{2} \left| D_{0}\psi \right|_{\C^{N}}^{2}-\frac{1}{2} \sum_{j=1}^{3}\left| D_{j}\psi \right|_{\C^{N}}^{2}
\end{equation}
where we have used (\ref{come-1-forme}), (\ref{prod-forme-c}) and the definition (\ref{diff-for-matter}) and
\begin{equation} \label{elle1-f}
\mathcal{L}_{1}(\Gamma):=-\frac{1}{2q^{2}}< F_{\Gamma
},F_{\Gamma }>_{M}=-\frac{1}{2q^{2}}< d\Gamma
+\Gamma \wedge \Gamma ,d\Gamma +\Gamma \wedge \Gamma >_{M},
\end{equation}
which using (\ref{curvature}) and (\ref{come-2-forme}) becomes
\begin{equation} \label{elle1}
\mathcal{L}_{1}(\Gamma)= \frac 12\, \sum_{j=1}^{3}\, \| F_{0j} \|^{2} -\frac 14\, \sum_{k,j=1}^{3}\, \| F_{kj} \|^{2}.
\end{equation}

A pure gauge field $\Gamma \in \Lambda^{1}(\mathbb{R}^{4},\mathfrak{g})$
by definition (see e.g. \cite{rub}, \cite{yangL}) is a critical point of the
action functional
\begin{equation*}
\mathcal{S}_{1}(\Gamma )=\int \mathcal{L}_{1}(\Gamma )\, dxdt;
\end{equation*}
whereas the action
\begin{equation*}
\mathcal{S}_{0}(\psi,\Gamma )=\int \mathcal{L}_{0}(\psi ,\Gamma ) \, dxdt;
\end{equation*}
formalises the reaction of the matter field $\psi$ to the gauge field $\Gamma$. A \textit{gauge theory ``with matter"} is then concerned with the functional
\begin{equation*}
\mathcal{S}_{01}(\psi ,\Gamma )=\int \left( \mathcal{L}_{0}+\mathcal{L}_{1}\right) \, dxdt.
\end{equation*}%
Since we are interested in the existence of solitary waves and solitons, we
add to the above lagrangian $\mathcal{L}_{0}+\mathcal{L}_{1}$ a nonlinear term $W$ and so we are concerned with the action
\begin{equation} \label{completa}
\mathcal{S}(\psi ,\Gamma )=\int \mathcal{L}(\psi ,\Gamma) \, dxdt,\qquad \mathcal{L}=\mathcal{L}_{0}+\mathcal{L}_{1}-W(\psi),
\end{equation}%
where $W:\mathbb{C}^{N}\rightarrow \mathbb{R}$ is a function which is
assumed to be $G$-invariant, namely $W(g\psi )=W(\psi )$ for all $g\in G$. In particular we assume that there exists a $C^{2}$ real function $f:\R^{+} \to \R$ such that
\begin{equation}  \label{4}
W(\psi ) = f(|\psi|), \qquad W'(\psi) = f'(|\psi|)\, \frac{\psi}{|\psi|} 
\end{equation}
and $f$ satisfies
\begin{itemize}
\item[(W1)] $f(s) \ge 0$ for all $s\ge 0$;

\item[(W2)] $f(0)=f'(0)=0$ and $f''(0) = m^2 >0$;

\item[(W3)] there exists $s_0 \in \R^+$ such that $f(s_0) < \frac{m^2}{2}\, s^2$

\item[(W4)] there exists $s_1 > s_{0} \in \R^+$ such that $f'(s_1) \ge s_{1}$
\end{itemize}

\begin{proposition} \label{deriv-azione}
The Euler-Lagrange equations relative to (\ref{completa}) have the following
form:
$$
\left\{
\begin{array}{ll}
D_{0}^{2}\psi -\sum_{j=1}^{3}D_{j}^{2}\psi + W'(\psi) = 0  & (E\psi) \\[0.2cm]
\Re \left\langle A , \sum_{j=1}^{3}D_{j} F_{0j} + q\, D_{0}\psi \cdot \psi^{*} \right\rangle =0 \quad \forall\, A \in \lie & (E\Gamma0) \\[0.2cm]
\Re \left\langle A , D_{0} F_{0j} - \sum_{\ell\not = j} D_{\ell} F_{\ell j} + q\, D_{j}\psi \cdot \psi^{*} \right\rangle =0 \quad \forall\, A \in \lie & (E\Gamma j), \quad j=1,2,3  
\end{array}
\right.
$$
where $D_{k}\psi \cdot \psi^{*}$ denotes the matrix multiplication between the column vector $D_{k}\psi$ and the row vector $\psi^{*}$, and we recall the definitions of the covariant derivatives (\ref{diff-for-matter}) for the matter field $\psi$ and (\ref{der-cov-gauge}) for $\lie$-valued functions.
\end{proposition}

\begin{proof}
To write the Euler-Lagrange equations relative to (\ref{completa}) we first consider the variation with respect to $\psi$, which gives
$$
d\mathcal{S}\left( \psi ,\Gamma \right) \left[ \left( \varphi ,0\right) \right] = - \Re \int < 
d_{\Gamma} \psi , d_{\Gamma }\varphi >_{M} \, dxdt = 
$$
$$
\Re \int \left((D_{0}\psi ,D_{0}\varphi)_{\C^{N}} - \sum_{j=1}^{3} ( D_{j}\psi ,D_{j}\varphi )_{\C^{N}} - ( W'(\psi), \varphi )_{\C^{N}} \right) \, dxdt = 
$$
$$
= \Re \int  (-D_{0}^{2}\psi + \sum_{j=1}^{3}D_{j}^{2}\psi - W'(\psi) ,\varphi )_{\C^{N}} \, dxdt
$$
where we have used the fact that matrices in $\lie$ are anti-Hermitian. We now consider the variations with respect to $\Gamma$ separately for each $\Gamma_{j}$ with $j=0,1,2,3$. For $j=0$ we have
$$
d\mathcal{S}\left( \psi ,\Gamma \right) \left[ \left( 0, B_{0} \right) \right] = \Re \int \left( ( qB_{0} \psi, D_{0} \psi )_{\C^{N}} - \sum_{j=1}^{3}\, \left\langle \partial_{j}B_{0} + q [\Gamma_{j},B_{0}], F_{0j} \right\rangle \right)\, dxdt=
$$
$$
= \Re \int \left\langle B_{0}, q\, D_{0}\psi \cdot \psi^{*} + \sum_{j=1}^{3} D_{j} F_{0j} \right\rangle\, dxdt
$$
where we have also used $Tr(ABC) = Tr(BCA)$ for any triple of square matrices $A,B,C$. The cases $j=1,2,3$ are similar with
$$
d\mathcal{S}\left( \psi ,\Gamma \right) \left[ \left( 0, B_{j} \right) \right] = 
$$
$$
= \Re \int \left( -( qB_{j} \psi, D_{j} \psi )_{\C^{N}} + \left\langle \partial_{0}B_{j} +q[\Gamma_{0},B_{j}] , F_{0j} \right\rangle - \frac 12 \sum_{\ell \not= j} \left\langle \partial_{\ell}B_{j} + q [\Gamma_{\ell},B_{j}], F_{\ell j} - F_{j \ell} \right\rangle \right)\, dxdt=
$$
$$
= \Re \int \left\langle B_{j}, - q\, D_{j}\psi \cdot \psi^{*} - D_{0} F_{0j} + \sum_{\ell \not= j}  D_{\ell} F_{\ell j} \right\rangle\, dxdt
$$
\end{proof}

Of fundamental importance in Lagrangian dynamics is the existence of conservation laws, obtained for example by Noether's Theorem under the existence of group actions which leave invariant the Lagrangian density. The Lagrangian density (\ref{completa}) is invariant for the action of the Poincar\'e group of Lorentz transformations of the Minkowski space $\R^{4}$, and for \textit{gauge transformations} of the fields $(\psi, \Gamma)$. A gauge transformation $\gau$ defined by a section 
$$
\R^{4} \ni x \mapsto \gau(x) \in G
$$ 
acts as 
\begin{equation} \label{g-act-mf}
(\gau\, \psi)(x) = \gau(x)\, \psi(x) \in \C^N
\end{equation}
on the matter field $\psi$, and by
\begin{equation} \label{g-act-gf}
\gau\, \Gamma = \sum_{j=0}^3 \, \tilde \Gamma_j \, dx^j \in \Lambda^1(\R^4,\lie)
\end{equation}
on the gauge field $\Gamma$, where
\begin{equation} \label{g-act-singgf}
\tilde \Gamma_j (x) = \gau(x)\, \Gamma_j(x) \, \gau^{-1}(x) - \frac 1q \left( \partial_j \gau(x) \right) \gau^{-1}(x) \in \lie
\end{equation}
By Noether's Theorem, if a Lagrangian density $\LL(t,x,u,\partial_t u, \nabla u)$ is invariant for the action of a one-parameter group $H=\set{h_{\lambda}}$, $\lambda \in \R$, then the function
\begin{equation} \label{int-primo}
I = \int_{\R^3}\,\, \left[ \frac{\partial \LL}{\partial (\partial_t u)}\, \left( \frac{\partial (h_{\lambda}(u))}{\partial \lambda} - \partial_t u \, \frac{\partial (h_{\lambda}(t))}{\partial \lambda} - \sum_{j=1}^3 \partial_j u\, \frac{\partial (h_{\lambda}(x_j))}{\partial \lambda} \right) + \LL\, \frac{\partial (h_{\lambda}(t))}{\partial \lambda} \right]_{\lambda = 0}\ dx
\end{equation}
is an integral of motion if all the fields fall off at infinity ``sufficiently rapid" (see e.g. \cite{gelf-fom}). The integral of motions obtained from the invariance under Lorentz transformations are energy, momentum, angular momentum and velocity of the ergocenter. The invariance under gauge transformations yields instead the integrals of motions known as \textit{(Noether's) charges}. 

\subsection{The Abelian case $G=U(1)$} \label{sec:ab-case}

We briefly consider the case $N=1$ and $G=U(1)$, obtaining Klein-Gordon-Maxwell equations as a simple case of the framework developed in the previous section. In this case the real Lie algebra $\lie = \ualg{1}= \imath \R$. We write the connection form $\Gamma \in \Lambda^{1}(\R^{4},\lie)$ as
$$
\Gamma = \sum_{j=0}^{3}\, \Gamma_{j}\, dx^{j} = - \imath \varphi\, dx^{0} + \imath A_{1}\, dx^{1} + \imath A_{2}\, dx^{2} + \imath A_{3}\, dx^{3}
$$
and with abuse of notation, we introduce also the notation of $\Gamma$ as a four-vector with components
\begin{equation} \label{gamma-4-vec-abel}
\Gamma = (-\varphi, \A) \quad \text{where} \quad \A=( A_{1}, A_{2}, A_{3})
\end{equation}
We also introduce the notation for the partial derivatives in the space-time components in $\R^{4}$, as $\partial_{0} = - \partial_{t}$,   $\partial_{1} = \partial_{x}$, $\partial_{2} = \partial_{y}$, $\partial_{3} = \partial_{z}$, with $\nabla \psi = (\partial_{x} \psi, \partial_{y} \psi, \partial_{z} \psi)$. The covariant derivatives then take the form
\begin{equation} \label{deriv-cov-max}
D_{0} \psi := \left( - \partial_t - \imath q \varphi \right) \psi \qquad D_{j} \psi :=  \left( \partial_{j} + \imath q A_{j} \right) \psi, \quad  j=1,2,3
\end{equation}
Using this notation we rewrite $\LL_{0}(\psi, \Gamma)$ of (\ref{elle0}) and $\LL_{1}(\Gamma)$ of (\ref{elle1}) as follows. For $\LL_{0}$ we get
\begin{equation} \label{elle-0-p-abel}
\LL_{0}(\psi, \varphi, \A) = \frac 1 2\, \left| \partial_{t} \psi + \imath\, q\, \varphi \psi \right|^{2} - \frac 1 2\, \left| \nabla \psi + \imath\, q\, \A \psi \right|^{2}
\end{equation}
In this case the components $F_{kj}$ defined in (\ref{curvature}) are complex numbers given by
$$
F_{0j} = - \imath \partial_{t} A_{j} + \imath \partial_{j} \varphi \qquad j=1,2,3
$$
$$
F_{kj} = \imath \partial_{k} A_{j} - \imath \partial_{j} A_{k} \qquad k,j=1,2,3
$$
for which $\| F_{kj} \|^{2} = |F_{kj}|^{2}$. It follows that
$$
\sum_{j=1}^{3}\, \| F_{0j}\|^{2} = \sum_{j=1}^{3}\, \left( \partial_{t} A_{j} - \partial_{j} \varphi \right)^{2} = \left| \partial_{t} \A - \nabla \varphi \right|^{2}
$$
$$
\sum_{k,j=1}^{3}\, \|F_{ij}\|^{2} = 2\ \left( |F_{12}|^{2} + |F_{23}|^{2} + |F_{31}|^{2} \right)= 2\ \left| \nabla \times \A \right|^{2}
$$
and
\begin{equation} \label{elle-1-p-abel}
\LL_{1}(\varphi, \A) =  \frac 1 2\,  \left| \partial_{t} \A - \nabla \varphi \right|^{2} - \frac 1 2\, \left| \nabla \times \A \right|^{2}
\end{equation}

To obtain the Klein-Gordon-Maxwell system of equations one can make the variations of ${\mathcal S}$ with respect to $\psi$, $\varphi$ and $\A$, or simply substitute the expressions for the covariant derivatives and the components $F_{kj}$ into the system $(E\psi),(E\Gamma_0),(E\Gamma_j)$ by taking into account that for $\lie = \imath \R$ these equations become 
$$
\left\{
\begin{array}{l}
D_{0}^{2}\psi -\sum_{j=1}^{3}D_{j}^{2}\psi + W'(\psi) = 0  \\[0.2cm]
\Im \left( \sum_{j=1}^{3}D_{j} F_{0j} + q\, D_{0}\psi \cdot \psi^{*} \right) =0 \\[0.2cm]
\Im \left( D_{0} F_{0j} - \sum_{\ell\not = j} D_{\ell} F_{\ell j} + q\, D_{j}\psi \cdot \psi^{*} \right) =0
\end{array}
\right.
$$
Hence we obtain
\begin{eqnarray}
& D^{2}_{0}\, \psi - \sum_{j=1}^{3}\, D^{2}_{j}\, \psi + W'(\psi)=0 \label{prima-kgm} \\[0.2cm]
& \nabla \cdot \left( \partial_{t} \A - \nabla \varphi \right) + q\, \Re(\imath\, \psi\, \partial_t \bar \psi) + q^2\, |\psi|^2 \, \varphi = 0  \label{seconda-kgm} \\[0.2cm]
& \partial_{t}  \left( \partial_{t} \A - \nabla \varphi \right)  + \nabla \times (\nabla \times \A) + q\, \Re (\imath\, \psi \, \nabla \bar \psi) + q^2\, |\psi |^2 \A= 0  \label{terza-kgm}
\end{eqnarray}

A useful approach to equations (\ref{prima-kgm})-(\ref{terza-kgm}) is to look for solutions $\psi(t,x) \in \C$ written in polar form, that is
\begin{equation} \label{polar-1}
\psi(t,x) = u(t,x) \, e^{\imath \, S(t,x)}, \qquad u\in \R^{+}, \ S\in \R/2\pi \Z
\end{equation}
Using notation (\ref{polar-1}), equation (\ref{prima-kgm}) splits in the equations
\begin{eqnarray}
& \partial_{t}^{2} u -\triangle u + \left[ |\nabla S + q\A|^{2} - \left( \partial_{t} S +q\varphi \right)^{2} \right] \, u + f'(u) =0 \label{prima-kgm-p1} \\[0.2cm]
& \partial_{t} \left[ \left( \partial_{t}S +q\varphi \right) u^{2} \right] - \nabla \cdot \left[ \left( \nabla S + q\A \right) u^{2} \right] =0 \label{prima-kgm-p2}
\end{eqnarray}
and (\ref{seconda-kgm}) and (\ref{terza-kgm}) become
\begin{eqnarray}
& \nabla \cdot \left( \partial_{t} \A - \nabla \varphi \right) + q\, (\partial_{t} S + q\varphi) \, u^{2}=0 \label{seconda-kgm-p} \\[0.2cm]
& \partial_{t}  \left( \partial_{t} \A - \nabla \varphi \right)  + \nabla \times (\nabla \times \A) + q\, (\nabla S + q \A)\, u^{2} =0  \label{terza-kgm-p}
\end{eqnarray}
Letting 
$$
\rho = \left( \partial_{t}S +q\varphi \right) u^{2} \qquad \mathbf{j} = - \left( \nabla S + q\A \right) u^{2}
$$
equation (\ref{prima-kgm-p2}) is the continuity equation for the electric charge density $\rho$. 

Finally we recall that the Lagrangian $\LL(\psi, \varphi, \A) = \LL_{0}(\psi, \varphi, \A) + \LL_{1}(\varphi, \A) - W(\psi)$ is invariant under the action of the Poincar\'e group and of the gauge transformation (\ref{g-act-mf}), (\ref{g-act-gf}), which in this case are of the form $\gau(x) = e^{\imath \, t(x)}$, $t(x) \in \R$, and transformation (\ref{g-act-singgf}) becomes
$$
\tilde \Gamma_j (x) = \Gamma_j -\frac \imath q \, \partial_j t(x)  
$$
The integrals of motion obtained by Noether's Theorem are energy, which using (\ref{polar-1}) takes the form
\begin{equation} \label{ener-kgm-p}
{\mathcal E} = \frac 12\ \int_{\R^3} \, \left[ (\partial_{t} u)^{2} + |\nabla u|^{2} + \frac{\rho^{2} + |{\mathbf j}|^{2}}{u^{2}} + 2W(u) + |\partial_{t} \A - \nabla \varphi|^{2} + |\nabla \times \A|^{2}  \right]\, dx 
\end{equation}
momentum, angular momentum which using (\ref{polar-1}) takes the form
\begin{equation} \label{angmom-kgm-p}
{\mathbf {\mathcal M}} = \int_{\R^{3}}\, {\mathbf x} \times \left[ \partial_{t}u \, \nabla u - \frac{\rho\, {\mathbf j}}{u^{2}} + \left( \partial_{t} \A + \nabla \varphi \right) \times (\nabla \times \A)   \right]\, dx
\end{equation}
velocity of the ergocenter and the Noether's charge, which is
\begin{equation} \label{charge-kgm-p}
{\mathcal Q} = \int_{\R^3}\, \rho \, dx = \int\, \left( \partial_{t}S +q\varphi \right) u^{2}\, dx
\end{equation}
and corresponds to the electric charge.

The existence of soliton and vortices solutions to equations (\ref{prima-kgm-p1})-(\ref{terza-kgm-p}) has been proved in \cite{bf,bf-vort} using the following ansatz
\begin{equation} \label{vortices-kgm}
\psi(t,x) = u(x)\, e^{\imath(\ell \vartheta(x)- \omega  t)}\ \ \text{and}\ \ (\varphi, \A) = (\varphi(x),\A(x))
\end{equation}
where $u\ge 0$, $\omega \in \R$, $\ell \in \Z$ and $\vartheta(x)$ is the longitude of the vector $\vec{r}=(x,y,z)$ in radial coordinates in $\R^3$. These solutions have non-vanishing matter angular momentum (see \cite{bf-vort})
\begin{equation} \label{mangmom-kgm-p}
{\mathbf {\mathcal M}}_{m} : = \int_{\R^{3}}\, {\mathbf x} \times \left[ \partial_{t}u \, \nabla u - \frac{\rho\, {\mathbf j}}{u^{2}} \right]\, dx = \int_{\R^{3}}\, u^{2}\, (\partial_{t}S + q\varphi)\, \left[ {\mathbf x} \times \left( \nabla S + q\A \right) \right]\, dx
\end{equation}
and hence are called vortices, when $\ell \not= 0$. 

Let $\hat{H}^{1}(\R^{3})$ denote the weighted Sobolev space of functions with norm
$$\| u \|^{2}_{\hat{H}^{1}} := \int\, \left[ |\nabla u|^{2} + \left( 1 + \frac{\ell^{2}}{r^{2}} \right)\, u^{2} \right]\, dx$$
for $\ell \in \Z$, and ${\mathcal D}^{1,2}$ denote the completion of $C^{\infty}_{0}(\R^{3})$ with respect to the inner product
$$(v|w)_{{\mathcal D}^{1,2}}:= \int\, \nabla v\cdot \nabla w\, dx$$
Then the following theorem holds
\begin{theorem}[\cite{bf-vort}] \label{main-bf}
Let $W$ satisfy (W1)-(W4). Then for all $\ell \in \Z$ there exists $q_0 >0$ such that for every $q \in (0,q_0)$ the system (\ref{prima-kgm-p1})-(\ref{terza-kgm-p}) admits a finite energy solution $(u,\omega,\varphi,\A)$ in the sense of distributions with: $u=u(\sqrt{x^2+y^2},z)\not\equiv 0$; $\omega >0$; $\varphi=\varphi(\sqrt{x^2+y^2},z)\not\equiv 0$; $\A = a(\sqrt{x^2+y^2},z)\, \nabla \vartheta$. In particular $u\in \hat{H}^{1}(\R^{3})$, $\varphi \in {\mathcal D}^{1,2}$ and $\A \in ({\mathcal D}^{1,2})^{3}$. Moreover, $\A \equiv 0$ if and only if $\ell = 0$.
\end{theorem}

We remark that a nonlinear term $W$ satisfying (W1)-(W4) is not of Higgs-type. This choice is justified by the following theorem proved in \cite{bf-arx}, whose proof we repeat here for completeness.

\begin{theorem}[\cite{bf-arx}] \label{no-higgs}
Let $W:\C\to \R$ be of the form $W(\psi) = f(|\psi|)$, where $f:\R^{+}\to \R$ is a $C^{2}$ function satisfying
\begin{itemize}
\item[(i)] $f(s)\ge 0$ for all $s\in \R^{+}$;
\item[(ii)] $f(0)>0$;
\item[(iii)] there exists $\bar s>0$ such that $f(\bar s)=0$.
\end{itemize}
Then system (\ref{prima-kgm-p1})-(\ref{terza-kgm-p}) admits no solutions of the form (\ref{vortices-kgm}) with $\omega, \ell\not= 0$ and finite energy.
\end{theorem}

\begin{proof}
The energy $\mathcal E$  (\ref{ener-kgm-p}) on functions of the form \eqref{vortices-kgm} writes
\begin{equation} \label{en-ansatz}
{\mathcal E}(u,\ell,\omega,\varphi,\A) = \frac 12\ \int_{\R^3} \, \Big[ |\nabla u|^{2} + (q\varphi -\omega)^{2}\, u^{2} + |\ell \nabla \vartheta +q\A|^{2}\, u^{2} + 2W(u) + |\nabla \varphi|^{2} + |\nabla \times \A|^{2}  \Big]\, dx 
\end{equation}
which is a sum of non-negative terms. Hence if energy is finite, all single terms are finite too. Hence arguing by contradiction, if there exists a finite energy solution to (\ref{prima-kgm-p1})-(\ref{terza-kgm-p}) of the form \eqref{vortices-kgm}, then finiteness of the term $\int\, W(u)\, dx$ implies that
$$\lim_{|x|\to \infty}\, u(x) = \bar s$$
by assumptions (i)-(iii) on $W$. From this and finiteness of energy it follows that
$$\int_{\R^3} \, |\ell \nabla \vartheta +q\A|^{2} u^{2}\, dx  < \infty$$
which implies that for any $\eps>0$ there exist $R,K>0$ such that
$$|\ell \nabla \vartheta(r,z) +q\A(r,z)| < \eps\qquad \forall\, r>R,\, |z|>K$$
where $r:=\sqrt{x^2+y^2}$. Hence
$$|\A(r,z)| > \frac 1q \left( |\ell \nabla \vartheta(r,z)|-\eps\right) = \frac 1q \left( \frac{|\ell|}{r}-\eps\right) \qquad \forall\, r>R,\, |z|>K$$
It follows that for $\eps$ small enough
$$\int_{\R^{3}} |\A(x)|^{6}\, dx \ge \int_{R}^{\infty} \int_{|z|>K}\, \left[\frac 1q \left( \frac{|\ell|}{r}-\eps\right)\right]^{6}\, r dr dz = \infty$$
Therefore $\A \not\in (L^{6}(\R^{3}))^{3}$, and then by Sobolev inequality
$$\int_{\R^{3}}\, |\nabla \A|^{2}\, dx = \infty$$
which contradicts the finiteness of the energy \eqref{en-ansatz} in the case $\ell\not= 0$. The same argument works for $\varphi$, and contradicts the finiteness of the energy in the case $\omega\not= 0$.
\end{proof}

\section{The simplest non-Abelian case $G=SU(2)$} \label{sec:nonab-case}

We now study a gauge theory on $\R^4$ with the matter field $\psi : \R^4 \to \C^2$ under the action of the gauge group $SU(2)$, with $\lie = \su{2}$, the real Lie algebra generated by $\imath$ times the Pauli matrices
$$
\tau_{1}:= \imath \sigma_{x} = \matrdpd{0}{\imath}{\imath}{0} \qquad
\tau_{2}:= \imath \sigma_{y} = \matrdpd{0}{1}{-1}{0} \qquad
\tau_{3}:= \imath \sigma_{z} = \matrdpd{\imath}{0}{0}{-\imath}
$$
By the properties of compact Lie groups, the exponential map
$$\exp : \su{2} \to SU(2)$$
is surjective and for each $g\in SU(2)$ there exists a triple $S=(S_1,S_2,S_3) \in \R^3$ with $\sum_{i=1}^3 S_i^2 \le \pi^2$ such that
$$g = \exp(S_1 \tau_1 + S_2 \tau_2 + S_3 \tau_3)$$
and it is unique when $\sum_{i=1}^3 S_i^2 < \pi^2$.
Given $(S_1,S_2,S_3) \in \R^3$ we introduce the notation
\begin{equation} \label{fasi-su2}
\SU(t,x) := S_1(t,x)\, \tau_1 + S_2(t,x)\, \tau_2 + S_3(t,x)\, \tau_3\in \su{2}, \qquad |\SU|^2:= |S|^2 = \sum_{i=1}^3 S_i^2
\end{equation}
and the operations
\begin{eqnarray}
& \partial_j \SU := \partial_j S_1(t,x)\, \tau_1 + \partial_j S_2(t,x)\, \tau_2 + \partial_j S_3(t,x)\, \tau_3 \label{der-pf} \\[0.2cm]
& \SU \times \tilde \SU := (S \times \tilde S)_1\, \tau_1 + (S \times \tilde S)_2\, \tau_2+ (S \times \tilde S)_3\, \tau_3 = -\frac 12\, [\SU,\tilde \SU] \label{lb-pf} \\[0.2cm]
& \SU \cdot \tilde \SU := S_1 \tilde S_1 + S_2 \tilde S_2 + S_3 \tilde S_3 = \frac 12 \langle \SU, \tilde \SU \rangle \label{tr-pf} \\[0.2cm]
& \SU \, \tilde \SU = - \SU \cdot \tilde \SU - \SU \times \tilde \SU \label{pr-pf}
\end{eqnarray}
where $[\cdot,\cdot]$ is the standard Lie bracket and in the last equation on the left hand side we use the usual matrix product. Finally for the gauge fields with abuse of notation we write
\begin{equation} \label{gf-pf}
\Gamma_{j} := \gamma_{j,1}\, \tau_{1} + \gamma_{j,2}\, \tau_{2} + \gamma_{j,3}\, \tau_{3}, \qquad j=0,1,2,3
\end{equation}
as in (\ref{fasi-su2}), and extend to $\Gamma_{j}$ the operations (\ref{der-pf}) and (\ref{lb-pf}). We then introduce the polar form for matter fields
\begin{equation}\label{polar-2}
\psi(t,x) = u(t,x) \, e^{\SU(t,x)}\, \psi_0, \qquad u\in \R^+,\, |\SU(t,x)| \le \pi
\end{equation}
for a fixed vector $\psi_0 \in \C^2$, $|\psi_{0}|_{\C^{2}}=1$. We now write the Lagrangian density $\LL_0$ as a function of the polar variables $(u,\SU)$ and their derivatives with respect to $\partial_{t} = -\partial_{0}$ and $\nabla = (\partial_{1},\partial_{2},\partial_{3})$.

\begin{lemma} \label{lem:derivata-esp}
For all $\SU \in \su{2}$ with regular functions $S_i(t,x)$, it holds
\begin{equation} \label{form-esp}
\partial_{j} \exp(\SU) = C(\SU, \partial_{j}\SU)\, \exp(\SU)
\end{equation}
with
$$
C(\SU, \partial_{j}\SU) := \partial_j \SU + \frac 12 \, (1-\cos 2)\, (\partial_j \SU \times \SU) + \frac 12\, (2-\sin 2)\, ((\partial_j \SU \times \SU) \times \SU) \in \su{2}
$$
\end{lemma}

\begin{proof}
We use the formula (see e.g. \cite{matrix})
\begin{equation} \label{feyn}
\frac{d}{dt} \exp(A(t)) = \sum_{k=0}^\infty\ \frac{1}{(k+1)!}\, (\text{ad} \, A(t))^k (A'(t))\, \exp(A(t))
\end{equation}
where
$$(\text{ad} \, A) (B) := [A,B]$$
By applying (\ref{feyn}) with $A(t) = \SU$ and using (\ref{lb-pf}) to obtain
$$[\SU, \partial_j \SU] = 2 (\partial_j \SU \times \SU)$$
$$\left[ \SU, [\SU, \partial_{j} \SU] \right] = \left[ \SU, 2 (\partial_j \SU \times \SU) \right] = 4\, \left(  (\partial_j \SU \times \SU)  \times \SU \right)$$
$$\left[ \SU, \left[ \SU, [\SU, \partial_{j} \SU] \right] \right] = 8\, \left( \left(  (\partial_j \SU \times \SU)  \times \SU \right) \times \SU \right) = -8\, (\partial_j \SU \times \SU)$$
and in general
$$\frac{1}{(k+1)!}\, (\text{ad} \, \SU)^k (\partial_{j}\SU) = \left\{ \begin{array}{ll} (-1)^{n-1} \frac{2^{2n-1}}{(2n)!}\, (\partial_j \SU \times \SU) & k=2n-1 \ge 1\\[0.2cm] (-1)^{n-1} \frac{2^{2n}}{(2n+1)!}\, \left(  (\partial_j \SU \times \SU)  \times \SU \right) & k=2n \ge 2
\end{array} \right.$$
Hence by (\ref{feyn})
$$
\partial_j \exp(\SU) = (\partial_{j} \SU) \exp(\SU) + \left( \sum_{n=1}^{\infty}\, (-1)^{n-1} \frac{2^{2n-1}}{(2n)!}\right) (\partial_j \SU \times \SU) \exp(\SU) +$$ $$+ \left( \sum_{n=1}^{\infty}\, (-1)^{n-1} \frac{2^{2n}}{(2n+1)!} \right) \left(  (\partial_j \SU \times \SU)  \times \SU \right) \exp(\SU)
$$
and the thesis follows.
\end{proof}

Using (\ref{form-esp}) we obtain for the covariant derivatives of the matter field $\psi$ in polar form (\ref{polar-2}) with respect to the gauge field $\Gamma = \sum\, \Gamma_{j}\, dx^{j}$
\begin{equation} \label{cov-der-polar}
D_{j} \left( u\, e^{\SU}\, \psi_{0} \right) = \partial_{j} \left( u\, e^{\SU}\, \psi_{0} \right) + q u\, \Gamma_{j} e^{\SU}\, \psi_{0} = \left[ \partial_{j} u + u\, C(\SU,\partial_{j}\SU) + q u\, \Gamma_{j} \right]\, e^{\SU}\, \psi_{0}
\end{equation}
We have then
$$| D_{j} \left( u\, e^{\SU}\, \psi_{0} \right) |_{\C^{2}}^{2} = \left( \left[ \partial_{j} u\, + u\, C(\SU,\partial_{j}\SU) + q u\, \Gamma_{j} \right]\, e^{\SU}\, \psi_{0} , \left[ \partial_{j} u + u\, C(\SU,\partial_{j}\SU) + q u\, \Gamma_{j} \right]\, e^{\SU}\, \psi_{0} \right)_{\C^{2}}= $$
$$= \left( \partial_{j} u\, e^{\SU}\, \psi_{0}, \partial_{j} u\, e^{\SU}\, \psi_{0} \right)_{\C^{2}} + \left( \left[ u\, C(\SU,\partial_{j}\SU) + q u\, \Gamma_{j} \right]\, e^{\SU}\, \psi_{0} , \left[ u\, C(\SU,\partial_{j}\SU) + q u\, \Gamma_{j} \right]\, e^{\SU}\, \psi_{0}\right)_{\C^{2}}$$
since the other terms cancel out using $ \left[ u\, C(\SU,\partial_{j}\SU) + q u\, \Gamma_{j} \right]^{*} = - \left[ u\, C(\SU,\partial_{j}\SU) + q u\, \Gamma_{j} \right]$, which holds for matrices in $\su{2}$. Moreover, since $|e^{\SU}\, \psi_{0}|_{\C^{2}} =1$ and $\SU^{2} = -|\SU|^{2}\, Id$, we obtain
\begin{equation} \label{norma-cov-der-polar}
| D_{j} \left( u\, e^{\SU}\, \psi_{0} \right) |_{\C^{2}}^{2} = |\partial_{j} u|^{2} + u^{2}\, \left| C(\SU,\partial_{j}\SU) + q\, \Gamma_{j} \right|^{2}
\end{equation}
Finally we obtain for $\LL_{0}$
\begin{equation} \label{l0-polar}
\LL_{0}(u,\SU,\Gamma) = \frac 12 |\partial_{t} u|^{2} - \frac 12 |\nabla u|^{2} + \frac 12 \, u^{2}\, \left[ \left| C(\SU,\partial_{t}\SU) - q\, \Gamma_{0} \right|^{2} - \sum_{j=1}^{3} \left| C(\SU,\partial_{j}\SU) + q\, \Gamma_{j} \right|^{2} \right]
\end{equation}
For the second part of the Lagrangian density $\LL_{1}$, we recall that since $F_{kj} \in \su{2}$ we have
$$
\| F_{kj} \|^{2} := Tr(F_{kj}^{*}\, F_{kj}) = - Tr(F_{kj}^{2}) = 2 |{\mathbf F}_{kj}|^{2} = 2 \left| \partial_{k} \Gamma_{j} - \partial_{j} \Gamma_{k} -2q\, (\Gamma_{k} \times \Gamma_{j}) \right|^{2}
$$
where for ${\mathbf F}_{kj}$ we have used notation (\ref{fasi-su2}), and for $\Gamma_{j}$ notation (\ref{gf-pf}) and (\ref{lb-pf}). Hence from (\ref{elle1}) we get
\begin{equation} \label{l1-polar}
\LL_{1}(\Gamma) = \sum_{j=1}^{3}\, \left| \partial_{t} \Gamma_{j} + \partial_{j} \Gamma_{0} + 2q\, (\Gamma_{0} \times \Gamma_{j})\right|^{2} - \frac 12\, \sum_{k,j=1}^{3}\, \left| \partial_{k} \Gamma_{j} - \partial_{j} \Gamma_{k} -2q\, (\Gamma_{k} \times \Gamma_{j}) \right|^{2}
\end{equation}

\subsection{The equations in polar form} \label{eq-su2-pf}

We now write equations $(E\psi), (E\Gamma0), (E\Gamma j)$ in terms of the polar form (\ref{polar-2}) for $\psi$ and using notation (\ref{fasi-su2}), (\ref{der-pf})-(\ref{gf-pf}). To write $(E\psi)$, we first use Lemma \ref{lem:derivata-esp} to write
$$
D_{j}^{2} \psi = D_{j} \left[ \left( \partial_{j} u + u C(\SU, \partial_{j} \SU) +q u \Gamma_{j} \right) \, e^{\SU} \psi_{0} \right] = 
$$
$$
=\Big[ \partial_{j} \left( \partial_{j} u + u C(\SU, \partial_{j} \SU) +q u \Gamma_{j} \right) + \left( \partial_{j} u + u C(\SU, \partial_{j} \SU) +q u \Gamma_{j} \right)\, C(\SU, \partial_{j} \SU) +
$$
$$
+ q\Gamma_{j} \left( \partial_{j} u + u C(\SU, \partial_{j} \SU) +q u \Gamma_{j} \right) \Big]\, e^{\SU} \psi_{0} =
$$
$$
= \Big[ \partial_{j}^{2} u + u \Big( C(\SU, \partial_{j} \SU)^{2} + 2q \Gamma_{j} C(\SU, \partial_{j} \SU) + q^{2}\, \Gamma_{j}^{2} \Big) + \frac 1u \partial_{j} \Big( \left( C(\SU, \partial_{j} \SU) +q\Gamma_{j} \right) \, u^{2} \Big) \Big]\, e^{\SU} \psi_{0} =
$$
$$
= \Big[ \partial_{j}^{2} u - u\, \Big| C(\SU, \partial_{j} \SU) +q\Gamma_{j} \Big|^{2} + u\, [q\Gamma_{j}, C(\SU, \partial_{j} \SU)]+  \frac 1u \partial_{j} \Big( \left( C(\SU, \partial_{j} \SU) +q\Gamma_{j} \right) \, u^{2} \Big) \Big]\, e^{\SU} \psi_{0}
$$
Hence we obtain
$$
\begin{array}{c}
\Big[ \partial_{0}^{2} u - \Delta u - u\, \Big| C(\SU, \partial_{0} \SU) +q\Gamma_{0} \Big|^{2} + \sum_{j=1}^{3}\, u\, \Big| C(\SU, \partial_{j} \SU) +q\Gamma_{j} \Big|^{2} + f'(u) + \\[0.2cm]
+ \frac 1u \partial_{0} \Big( \left( C(\SU, \partial_{0} \SU) + q\Gamma_{0} \right) \, u^{2} \Big) + u\, [q\Gamma_{0}, C(\SU, \partial_{0} \SU)] + \\[0.2cm]
- \sum_{j=1}^{3}\, \frac 1u \partial_{j} \Big( \left( C(\SU, \partial_{j} \SU) +q\Gamma_{j} \right) \, u^{2} \Big) - u\, [q\Gamma_{j}, C(\SU, \partial_{j} \SU)] \Big]\, e^{\SU} \psi_{0} =0
\end{array}
$$
We remark that in the first row all terms in square brackets are assumed to multiply the identity matrix, which is not in $\su{2}$, whereas in the second and third lines all terms are in $\su{2}$. Hence using the covariant derivative \eqref{der-cov-gauge} and letting $\partial_{t} = -\partial_{0}$ we obtain the system
\begin{eqnarray} 
& \partial_{t}^{2} u - \Delta u - u\, \Big| C(\SU, \partial_{t} \SU) -q\Gamma_{0} \Big|^{2} + \sum_{j=1}^{3}\, u\, \Big| C(\SU, \partial_{j} \SU) +q\Gamma_{j} \Big|^{2} + f'(u) = 0 \label{epsi-pf}  \\[0.2cm]
& D_{0} \Big( \left( C(\SU, \partial_{0} \SU) +q\Gamma_{0} \right) \, u^{2} \Big) - \sum_{j=1}^{3}\, D_{j} \Big( \left( C(\SU, \partial_{j} \SU) +q\Gamma_{j} \right) \, u^{2} \Big) = 0 \label{epsi-pf-cont}
\end{eqnarray}
Notice that equation \eqref{epsi-pf-cont} is the continuity equation in terms of covariant derivatives for the matter current (see \eqref{egamma0-pf} and \eqref{egammaj-pf}).

We now consider $(E\Gamma0)$ with $A= \tau_{k}$, $k=1,2,3$. Since $D_{j}F_{0j}$ is in $\su{2}$ the first part simply reads
$$
\Re \left\langle \tau_{k}, D_{j}F_{0j} \right\rangle = \sum_{m=1}^{3}\, ( D_{j}F_{0j} )_{m}\, \Re \left\langle \tau_{k}, \tau_{m} \right\rangle = 2 ( D_{j}F_{0j} )_{k}
$$
The second part of $(E\Gamma0)$ involves $D_{0}\psi\cdot \psi^{*}$ and needs more attention. In general we can write
$$
\psi(t,x) = u(t,x) e^{\SU(t,x)}\psi_{0} = u(t,x) \left( \begin{array}{c} z \\ w \end{array} \right)
$$
with $z,w \in \C$ and $|z|^{2}+|w|^{2}=1$. Hence for all $j=0,1,2,3$ we can write
$$
\begin{array}{rl}
D_{j}\psi \cdot \psi^{*} & =  \Big[ u \partial_{j} u + u^{2} (C(\SU, \partial_{j} \SU) + q\Gamma_{j}) \Big]\, e^{\SU} \psi_{0} \cdot \left( e^{\SU} \psi_{0} \right)^{*} \\[0.3cm] & = \Big[ u \partial_{j} u + u^{2} (C(\SU, \partial_{j} \SU) + q\Gamma_{j}) \Big]\, \left( \begin{array}{cc} |z|^{2} & z \bar{w} \\ \bar{z} w & |w|^{2} \end{array} \right)
\end{array}
$$
We now compute
$$
\Re \left\langle \tau_{k}, \left( \begin{array}{cc} |z|^{2} & z \bar{w} \\ \bar{z} w & |w|^{2} \end{array} \right) \right\rangle =0 \qquad \forall\, k=1,2,3
$$
from which it follows that
$$
\Re \left\langle \tau_{k}, D_{j}\psi \cdot \psi^{*} \right\rangle = u^{2}\, \sum_{m=1}^{3}\, (C(\SU, \partial_{j} \SU) + q\Gamma_{j})_{m} \, \Re \left\langle \tau_{k}, \tau_{m}\, \left( \begin{array}{cc} |z|^{2} & z \bar{w} \\ \bar{z} w & |w|^{2} \end{array} \right) \right\rangle =
$$
$$
= u^{2} (C(\SU, \partial_{j} \SU) + q\Gamma_{j})_{k} \, \, tr  \left( \begin{array}{cc} |z|^{2} & z \bar{w} \\ \bar{z} w & |w|^{2} \end{array} \right) = u^{2} (C(\SU, \partial_{j} \SU) + q\Gamma_{j})_{k}
$$
Hence we finally get
\begin{equation} \label{egamma0-pf}
2 \, \sum_{j=1}^{3}\, D_{j} F_{0j} - q\, u^{2} [C(\SU, \partial_{t} \SU) - q\Gamma_{0}] =0
\end{equation}
The same argument works for $(E\Gamma j)$ and we get
\begin{equation} \label{egammaj-pf}
2\, D_{0} F_{0j} - 2\, \sum_{\ell\not= j}\, D_{\ell} F_{\ell j} + q\, u^{2} [C(\SU, \partial_{j} \SU) + q\Gamma_{j}] = 0, \qquad j=1,2,3 
\end{equation}

\subsection{The energy}\label{sec:energy} 
The conservation law of energy follows from the invariance of the Lagrangian under the action of the one-parameter group of time translations
$$
h_\lambda (t,x, u, \SU, \Gamma) = (t+\lambda, x, u, \SU, \Gamma)
$$
which yields by (\ref{int-primo})
\begin{equation}\label{en-polar-su2-1}
{\mathcal E} = \int_{\R^3}\, \left[ \frac{\partial \LL}{\partial (\partial_t u)} \, \partial_t u + \sum_{m=1}^3\, \frac{\partial \LL}{\partial (\partial_t S_m)} \, \partial_t S_m + \sum_{j=0}^3 \sum_{m=1}^3\, \frac{\partial \LL}{\partial (\partial_t \gamma_{j,m})} \, \partial_t \gamma_{j,m} - \LL \right]\, dx
\end{equation}
We now compute the different terms. The first two terms depend only on $\LL_0$ (\ref{l0-polar}) and by making the trivial computations using (\ref{der-ut})-(\ref{der-S3t}) we get
$$
\frac{\partial \LL}{\partial (\partial_t u)} \, \partial_t u + \sum_{m=1}^3\, \frac{\partial \LL}{\partial (\partial_t S_m)} \, \partial_t S_m - \LL_0 =
$$
$$
= \frac 12 |\partial_{t} u|^{2} + \frac 12 |\nabla u|^{2} + \frac 12 \, u^{2}\, \left[ \left| C(\SU,\partial_{t}\SU) \right|^2 - q^2\, \left| \Gamma_{0} \right|^{2} + \sum_{j=1}^{3} \left| C(\SU,\partial_{j}\SU) + q\, \Gamma_{j} \right|^{2} \right]
$$
Finally the third term only depends on $\LL_1$ (\ref{l1-polar}) and in particular we remark that $\LL_1$ does not depend on $\partial_t \gamma_{0,m}$. Using (\ref{der-gammajmt}) we get
$$
\sum_{j=0}^3 \sum_{m=1}^3\, \frac{\partial \LL_1}{\partial (\partial_t \gamma_{j,m})} \, \partial_t \gamma_{j,m} -\LL_1 = \sum_{j=1}^3\, \sum_{m=1}^3\, 2\, \Big( \partial_{t}\gamma_{j,m} + \partial_{j} \gamma_{0,m} +2q\, (\Gamma_{0} \times \Gamma_{j})_{m} \Big)  \, \partial_t \gamma_{j,m} -\LL_1 = 
$$
$$
= \sum_{j=1}^3 \left\langle  \partial_t \Gamma_{j} + \partial_j \Gamma_{0} +2q\, \Gamma_0 \times \Gamma_j, \partial_t \Gamma_{j} \right\rangle -\LL_1 =
$$
$$
= \sum_{j=1}^{3}\, \left| \partial_{t} \Gamma_{j} + \partial_{j} \Gamma_{0} + 2q\, (\Gamma_{0} \times \Gamma_{j})\right|^{2} + \frac 12\, \sum_{k,j=1}^{3}\, \left| \partial_{k} \Gamma_{j} - \partial_{j} \Gamma_{k} -2q\, (\Gamma_{k} \times \Gamma_{j}) \right|^{2} + 
$$
$$
+ \sum_{j=1}^3 \, \left\langle F_{0j} ,  \partial_j \Gamma_{0} + 2q\,  \Gamma_0 \times \Gamma_j \right\rangle
$$
where we have used (\ref{tr-pf}). Putting together all the terms we obtain that the density of energy is
\begin{equation}\label{ff-de}
\begin{array}{c}
E = \frac 12 |\partial_{t} u|^{2} + \frac 12 |\nabla u|^{2} + \frac 12 \, u^{2}\, \left[ \left| C(\SU,\partial_{t}\SU) - q\, \Gamma_{0} \right|^{2} + \sum_{j=1}^{3} \left| C(\SU,\partial_{j}\SU) + q\, \Gamma_{j} \right|^{2} \right] +\\[0.2cm]
+ \sum_{j=1}^{3}\, \left| \partial_{t} \Gamma_{j} + \partial_{j} \Gamma_{0} + 2q\, (\Gamma_{0} \times \Gamma_{j})\right|^{2} + \frac 12\, \sum_{k,j=1}^{3}\, \left| \partial_{k} \Gamma_{j} - \partial_{j} \Gamma_{k} -2q\, (\Gamma_{k} \times \Gamma_{j}) \right|^{2} +\\[0.2cm]
+ u^{2}\, \left[ \frac q2\, \left\langle C(\SU, \partial_t \SU), \Gamma_{0} \right\rangle - q^2\, \left| \Gamma_{0} \right|^{2}  \right] + \sum_{j=1}^3 \, \left\langle F_{0j} ,  \partial_j \Gamma_{0} + 2q\,  \Gamma_0 \times \Gamma_j \right\rangle
\end{array}
\end{equation}

\subsection{Existence of solitary waves and vortices} \label{sec:exist}

We now introduce a particular ansatz to find solitary waves solutions for equations (\ref{epsi-pf}), \eqref{epsi-pf-cont}, (\ref{egamma0-pf}) and (\ref{egammaj-pf}). We restrict ourselves to solutions with matter field written in polar form as
\begin{equation} \label{solit-su2}
\psi(t,x,y,z) = u( \sqrt{x^{2}+y^{2}}, z)\, e^{S(t,x)\, \tau_{m}}\, \psi_{0}, \qquad u\in \R^{+},\ m=1,2,3,\ |S(t,x)|\le \pi
\end{equation}
with $S(t,x) = \ell \vartheta(x) - \omega \, t$, where $\omega \in \R$, $\ell \in \Z$ and $\vartheta(x)$ is the longitude of the vector $\vec{r}=(x,y,z)$ in radial coordinates in $\R^3$. For the gauge field we assume analogously that 
\begin{equation} \label{campo-su2}
\Gamma_{0} = \gamma_{0}( \sqrt{x^{2}+y^{2}}, z)\ \tau_{m}, \qquad \left( \begin{array}{c} \Gamma_{1} \\ \Gamma_{2} \\ \Gamma_{3} \end{array}\right) = \gamma( \sqrt{x^{2}+y^{2}}, z)\, \nabla \vartheta \ \tau_{m}
\end{equation}
In particular we look for the radial profile $u$ of the matter field and the gauge field to be independent on time. We now substitute (\ref{solit-su2}) and (\ref{campo-su2}) in the equations, using
$$
\triangle \vartheta = \nabla \vartheta \cdot \nabla u = 0,\qquad C(\SU, \partial_{j} \SU) = \partial_{j} S(t,x)\, \tau_{m}
$$
$$
F_{0j} = -\partial_{j} \gamma_{0}(x)\, \tau_{m}, \qquad F_{\ell j} = \Big[ \nabla \times \gamma(x) \nabla \vartheta \Big]_{k}\, \tau_{m}
$$
with $j, \ell, k=1,2,3$, and $\ell\not= j \not= k$ such that $(\ell, j, k)$ is obtained by an even permutation from $(1,2,3)$. We find the following equations for the variables $(u, \omega, \ell, \gamma_{0}, \gamma)$, with equation \eqref{epsi-pf-cont} identically satisfied,
\begin{eqnarray}
& -\triangle u(x) + \Big[ | (\ell + q\gamma(x)) \nabla \vartheta |^{2} - (\omega+q\gamma_{0}(x))^{2} \Big] \, u + f'(u) =0 \label{matter-su2} \\[0.2cm]
& -2\, \triangle \gamma_{0}(x) + q\, (\omega + q\gamma_{0}(x))\, u^{2} =0 \label{gauss-su2} \\[0.2cm]
& 2\, \nabla \times \left( \nabla \times \gamma(x) \nabla \vartheta \right) + q\, (\ell + q\gamma(x))\, u^{2}\, \nabla \vartheta =0 \label{rotore-su2}
\end{eqnarray}
which are very similar to (\ref{prima-kgm-p1}), (\ref{seconda-kgm-p}) and (\ref{terza-kgm-p}) if we let $u(t,x) = u(t)$, $S(t,x) = \ell \vartheta(x) - \omega t$, $\varphi(t,x) = -\gamma_{0}(x)$ and $\A(t,x) = \gamma(x) \nabla \vartheta(x)$.

\begin{theorem} \label{main-su2}
Let $W$ satisfy (W1)-(W4). Then for all $\ell \in \Z$ there exists $q_0 >0$ such that for every $q \in (0,q_0)$ the system (\ref{matter-su2})-(\ref{rotore-su2}) admits a finite energy solution $(u,\omega,\ell,\gamma_{0},\gamma)$ in the sense of distributions with: $u=u(\sqrt{x^2+y^2},z)\not\equiv 0$; $\omega >0$; $\gamma_{0}=\gamma_{0}(\sqrt{x^2+y^2},z)\not\equiv 0$; $\gamma= \gamma(\sqrt{x^2+y^2},z)$. Moreover, $\gamma \equiv 0$ if and only if $\ell = 0$.
\end{theorem}

\begin{proof}
By letting $v = \frac{1}{\sqrt{2}}\, u$ and $\tilde f(s) = \frac{1}{2}\, f(\sqrt{2} u)$,  we have that $(u,\omega,\gamma_{0},\gamma)$ is a solution of the system (\ref{matter-su2})-(\ref{rotore-su2}) if and only if $(v,\omega,\gamma_{0},\gamma)$ is a solution of
$$
\begin{array}{c}
-\triangle v + \Big[ | (\ell + q\gamma(x)) \nabla \vartheta |^{2} - (\omega+q\gamma_{0}(x))^{2} \Big] \, v + \tilde f'(v) =0 \\[0.2cm]
-\triangle \gamma_{0} + q\, (\omega + q\gamma_{0}(x))\, v^{2} =0 \\[0.2cm]
\nabla \times \left( \nabla \times \gamma \nabla \vartheta \right) + q\, (\ell + q\gamma(x))\, v^{2}\, \nabla \vartheta =0
\end{array}
$$
It is immediate to verify that $\tilde f$ satisfies assumptions (W1)-(W4), by choosing $\tilde s_{0} = \frac{1}{\sqrt{2}}\, s_{0}$ and $\tilde s_{1} = \frac{1}{\sqrt{2}}\, s_{1}$ in (W3) and (W4). Hence we can apply Theorem \ref{main-bf} to obtain the existence.

Finally, the integral of the energy density \eqref{ff-de} for the solutions of form \eqref{solit-su2} and \eqref{campo-su2} writes
$$
{\mathcal E}(u,\omega,\ell,\gamma_{0},\gamma)=\int_{\R^{3}} \left[ \frac 12 |\nabla u|^{2} + \frac 12 u^{2} \left[ \omega^{2} - q^{2}\, \gamma_{0}^{2} + \frac{(\ell + q\, \gamma)^{2}}{r^{2}} \right] - |\nabla \gamma_{0}|^{2} + |\nabla \times \gamma \nabla \vartheta|^{2} \right] \, dx
$$
The finiteness of the energy follows again by Theorem \ref{main-bf}.
\end{proof}

\appendix

\section{Derivatives of the Lagrangian density}

\noindent Let the Lagrangian density 
$$
\LL(u,\SU,\Gamma) = \LL_{0}(u,\SU,\Gamma) + \LL_{1}(\Gamma) - W(u)
$$
be expressed in terms of the variables $(u,\SU)$ for the matter field and $\Gamma = (\Gamma_{j})$, $j=0,1,2,3$, with $\Gamma_{j}= \sum_{m}\, \gamma_{j,m} \tau_{m}$, for the gauge field, as in (\ref{l0-polar}) and ({\ref{l1-polar}). In order to obtain the densities of the first integrals for equations (\ref{epsi-pf})-(\ref{egamma0-pf})-(\ref{egammaj-pf}), we need to compute the derivatives of $\LL$ with respect to $\partial_{t} u$, $\partial_{t} S_{m}$, $\partial_{t} \gamma_{j,m}$, for $m=1,2,3$ and $j=0,1,2,3$, as shown in (\ref{int-primo}). We collect here the derivatives.
\begin{equation} \label{der-ut}
\frac{\partial \LL}{\partial(\partial_{t}u)} = \partial_{t} u 
\end{equation}
\begin{equation} \label{der-S1t}
\frac{\partial \LL}{\partial(\partial_{t}S_{1})} = \begin{array}{l} u^{2}\, \Big[ \left( C(\SU, \partial_{t}\SU) - q\, \Gamma_{0} \right)_{1} \left( 1 - \frac 12 (2-\sin 2) (S_{2}^{2}+S_{3}^{2}) \right) + \\[0.2cm] + \left( C(\SU, \partial_{t}\SU) - q\, \Gamma_{0} \right)_{2} \left( \frac 12 (2-\sin 2) S_{1}S_{2} - \frac 12 (1-\cos 2) S_{3} \right) + \\[0.2cm] + \left( C(\SU, \partial_{t}\SU) - q\, \Gamma_{0} \right)_{3} \left( \frac 12 (2-\sin 2) S_{1}S_{3} + \frac 12 (1-\cos 2) S_{2} \right) \Big] \end{array}
\end{equation}
\begin{equation} \label{der-S2t}
\frac{\partial \LL}{\partial(\partial_{t}S_{2})} = \begin{array}{l} u^{2}\, \Big[ \left( C(\SU, \partial_{t}\SU) - q\, \Gamma_{0} \right)_{1} \left( \frac 12 (2-\sin 2) S_{1}S_{2} + \frac 12 (1-\cos 2) S_{3} \right) + \\[0.2cm] + \left( C(\SU, \partial_{t}\SU) - q\, \Gamma_{0} \right)_{2} \left( 1 - \frac 12 (2-\sin 2) (S_{1}^{2}+S_{3}^{2}) \right)  + \\[0.2cm] + \left( C(\SU, \partial_{t}\SU) - q\, \Gamma_{0} \right)_{3} \left( \frac 12 (2-\sin 2) S_{2}S_{3} - \frac 12 (1-\cos 2) S_{1} \right)  \Big] \end{array}
\end{equation}
\begin{equation} \label{der-S3t}
\frac{\partial \LL}{\partial(\partial_{t}S_{3})} = \begin{array}{l} u^{2}\, \Big[ \left( C(\SU, \partial_{t}\SU) - q\, \Gamma_{0} \right)_{1} \left( \frac 12 (2-\sin 2) S_{1}S_{3} - \frac 12 (1-\cos 2) S_{2} \right) + \\[0.2cm] + \left( C(\SU, \partial_{t}\SU) - q\, \Gamma_{0} \right)_{2} \left( \frac 12 (2-\sin 2) S_{2}S_{3} + \frac 12 (1-\cos 2) S_{1} \right)   + \\[0.2cm] + \left( C(\SU, \partial_{t}\SU) - q\, \Gamma_{0} \right)_{3} \left( 1 - \frac 12 (2-\sin 2) (S_{1}^{2}+S_{2}^{2}) \right)   \Big] \end{array}
\end{equation}
\begin{equation} \label{der-gammajmt}
\frac{\partial \LL}{\partial(\partial_{t}\gamma_{j,m})} = 2\, \Big( \partial_{t}\gamma_{j,m} + \partial_{j} \gamma_{0,m} +2q\, (\Gamma_{0} \times \Gamma_{j})_{m} \Big), \qquad j,m=1,2,3
\end{equation}

\end{document}